\documentclass[11pt]{article}

\usepackage[a4paper,margin=0.9in]{geometry}
\usepackage[T1]{fontenc}
\usepackage[utf8]{inputenc}
\IfFileExists{lmodern.sty}{\usepackage{lmodern}}{}
\usepackage{microtype}
\usepackage{amsmath,amssymb,amsthm,mathtools,bm}
\usepackage{booktabs,array}
\usepackage{graphicx}
\graphicspath{{./}{../}}
\usepackage{caption}
\usepackage{float}
\usepackage{tikz}
\usetikzlibrary{arrows.meta,positioning,calc,decorations.pathreplacing}
\usepackage{enumitem}
\usepackage{xcolor}
\usepackage{hyperref}
\usepackage[nameinlink,capitalise,noabbrev]{cleveref}

\definecolor{linkblue}{RGB}{18,73,127}
\definecolor{protectblue}{RGB}{34,104,165}
\definecolor{protectred}{RGB}{181,31,46}
\definecolor{protectgreen}{RGB}{32,126,96}
\definecolor{mutedgray}{RGB}{110,116,124}
\hypersetup{
 colorlinks=true, linkcolor=linkblue, citecolor=linkblue, urlcolor=linkblue,
 pdftitle={Symmetry-Fixed Holonomies and Spectral Isolation in Two-Cycle Photonic Geometries},
 pdfauthor={Michel Planat}
}

\setlist[itemize]{leftmargin=1.5em,itemsep=2pt,topsep=4pt}
\setlist[enumerate]{leftmargin=1.8em,itemsep=3pt,topsep=4pt}
\newtheorem{theorem}{Theorem}[section]
\newtheorem{proposition}[theorem]{Proposition}

\newtheorem{remark}[theorem]{Remark}

\newcommand{\R}{\mathbb{R}}
\newcommand{\C}{\mathbb{C}}
\newcommand{\Z}{\mathbb{Z}}
\newcommand{\Q}{\mathbb{Q}}
\newcommand{\U}{\mathrm{U}}
\newcommand{\SU}{\mathrm{SU}}
\newcommand{\ii}{\mathrm{i}}
\newcommand{\dd}{\mathrm{d}}
\newcommand{\cH}{\mathcal{H}}

\newcommand{\Aut}{\operatorname{Aut}}
\newcommand{\diag}{\operatorname{diag}}
\newcommand{\dist}{\operatorname{dist}}
\newcommand{\Hess}{\operatorname{Hess}}

\title{\textbf{Symmetry-Fixed Holonomies and Spectral Isolation\\ in Two-Cycle Photonic Geometries}\\[0.35em]
\large A Square Parent Manifold for a Qubit and a Hexagonal Qutrit Manifold}
\author{Michel Planat\\
\small  Universit\'e Marie et Louis Pasteur, Institut FEMTO-ST, CNRS\\
\small 25000 Besan\c{c}on, France\\
\small \href{mailto:michel.planat@femto-st.fr}{michel.planat@femto-st.fr}}
\date{}

\begin{document}
\maketitle

\begin{abstract}
A system with two periodic directions carries two commuting holonomies \(a=(u,v)\in\R^2/\Z^2\). We determine their distinguished values while separating lattice, arithmetic, and observable effects. Maximizing the lowest twisted eigenvalue places \(a\) at a deep hole of the momentum lattice. For every rectangular torus the maximizer is antiperiodic, so complex multiplication is sufficient for torsion optima but not necessary. Let \(G_\tau^{-1}\) be the dual metric and \(D_\tau(a)\) the normalized zeta determinant of the twisted Laplacian. At the rotation-fixed deep holes of the square and hexagonal lattices, symmetry gives the exact determinant response \(-\operatorname{Hess}_a\log D_\tau=2\pi(\Im\tau)G_\tau^{-1}\). With spectral wavenumber \(\kappa=2\pi\), the lowest manifolds are fourfold and threefold, with gaps \(2\kappa^2\) and \(4\kappa^2/3\); phase errors split them linearly while their centroids remain stationary. We then give a finite-device realization: an \(8\times8\) microring lattice closed by two phase-controlled seams. At a reported coupling scale of \(16\) GHz, its exact square-lattice spectrum has a \(17.32\) GHz shell gap and a \(1.92\) GHz doublet separation for a \(0.1\) holonomy error; a triangular-link configuration gives a threefold qutrit manifold with an \(18.11\) GHz gap. This is a quantitative spectroscopy proposal, not a claim of topological protection or a completed device.
\end{abstract}



\section{Introduction}

Photonic quantum information requires a low-dimensional subspace that is controllable, measurable, and well separated from unwanted modes.  Path and polarization dual-rail qubits are standard in linear-optical quantum computing \cite{KokEtAl2007}; temporal and spectral modes provide another complete modal framework \cite{BrechtEtAl2015}, within the broader spatial, temporal, frequency, and polarization mode description of quantum light \cite{FabreTreps2020}.  Integrated photonics has also generated, manipulated, and measured path-encoded qutrits with three-dimensional multiports \cite{LuEtAl2020}.  Synthetic gauge fields and synthetic dimensions have made the associated phases directly programmable \cite{FangYuFan2012,MittalEtAl2014,YuanShiFan2016,DuttEtAl2020,OzawaEtAl2019,HafeziEtAl2011}.  Two experimental ingredients are especially relevant here: a \(10\times10\) silicon site-ring array was characterized with a tunnelling rate near \(16\) GHz and an intrinsic field-decay rate near \(1\) GHz (an intensity full width near \(2\) GHz in the convention of that work) \cite{HafeziEtAl2013}, while integrated lithium-niobate ring cavities have demonstrated phase-programmable long-range hopping and time-resolved synthetic-band spectroscopy \cite{DinhEtAl2024}.  \Cref{sec:device} combines these established ingredients into a concrete twisted-torus proposal; it does not claim that the required two seam closures have already been fabricated.

In such a system with two independent periodic directions the control variables are two commuting holonomies,
\begin{equation}
 a=(u,v)\in\R^2/\Z^2 ,
 \label{eq:holonomy-variables}
\end{equation}
and the question is which \(a\) are spectrally distinguished.  We answer three sharper questions about the antiperiodic point \(a=(1/2,1/2)\).

\begin{enumerate}
 \item \textbf{What selects it?}  Gap maximization alone selects it on every rectangular lattice, independently of aspect ratio and of any arithmetic property of \(\tau\).  Complex multiplication remains a correct sufficient condition for torsion optima, but the location of the optimum carries no information about it.
 \item \textbf{Where does arithmetic enter structurally?}  Through the unit group of the lattice.  Only two moduli admit units of order \(>2\).  At their rotation-fixed deep holes the current vanishes, while the extra rotational symmetry forces the quadratic response into the closed form \(\mathcal K=2\pi(\Im\tau)G_\tau^{-1}\).  Current cancellation alone is not exceptional: order-two inversion already forces it at every two-torsion point.
 \item \textbf{What is measurable, and what does it buy?}  We give exact level splittings and centroid curvatures, which are resonance-frequency observables, and we state the precise sense in which an isolated manifold assists a dual-rail qubit (square) or a qutrit (hexagonal).
\end{enumerate}

\begin{remark}[Global holonomy is not plaquette flux]\label{rem:not-hofstadter}
The variables \eqref{eq:holonomy-variables} are Wilson loops of a \emph{flat} connection around the two cycles, equivalently twisted boundary conditions.  They are not magnetic flux per plaquette: a flat connection has zero curvature, produces no magnetic translation algebra, and no Hofstadter butterfly.  We use ``antiperiodic holonomy point'' for \(a=(1/2,1/2)\), reserving ``half flux'' as an informal synonym.  The point nevertheless uses the same informal phase language as the \(\pi\)-flux sectors of magnetic lattice models \cite{Hofstadter1976}, and it is the half-quantum twist familiar from Aharonov--Bohm and persistent-current physics \cite{AharonovBohm1959,ButtikerImryLandauer1983}.  We record the analogy but do not use magnetic-flux results in any proof.
\end{remark}

\subsection*{Geometric motivation, and its limits}

The two-cycle structure has an appealing origin.  The curl spectrum on a round sphere is explicit \cite{Bar2019,Folland1989}, and related linked and torus-knotted free-space fields are developed in \cite{IrvineBouwmeester2008,KediaEtAl2013}.  More specifically, on the Einstein cylinder a source-free Maxwell solution aligned with the Clifford parallels of \(S^3\) is carried by conformal equivalence to finite-energy knotted radiation in Minkowski spacetime \cite{KopinskiNatario2017}.  Removing two Hopf fibres leaves a complement homeomorphic to \(T^2\times[0,1]\), with \(\pi_1=\Z^2\) and two commuting holonomies.

We state what this does and does not supply.  It supplies the control-variable count and a picture of two linked cycles.  It does \emph{not} determine the modulus \(\tau\) of the boundary torus, which depends on the tubular-neighbourhood radius and is a free design parameter, and no vector Maxwell mode on a finite-core link complement is solved here.  We therefore treat \(\tau\) as a design parameter and label the solvable models as benchmarks.  \Cref{tab:status}, the first display item, summarizes the logical status of the results and proposals before the technical development begins.

\begin{table}[ht]
\centering
\caption{Logical status of the statements in this paper.}
\label{tab:status}
\begin{tabular}{p{0.25\textwidth}p{0.65\textwidth}}
\toprule
Status & Statement \\
\midrule
Proved here &
Rectangular tori have their gap optimum at the antiperiodic point for every aspect ratio (\Cref{prop:rectangular});
CM implies rational, hence torsion, optima (\Cref{prop:rational});
the fixed sets of the quarter-turn and the \(120^\circ\) rotation are the origin together with the deep holes (\Cref{thm:units});
\(\Delta_\zeta\log D_\tau=-4\pi/\Im\tau\) and hence the exact response \(\mathcal K=2\pi(\Im\tau)G_\tau^{-1}\) (\Cref{thm:response});
\(D_\ii(1/2,1/2)=2\) and \(D_\rho(1/3,2/3)=\sqrt3\) (\Cref{prop:square-value,prop:hex-value});
exact linear splitting and centroid curvature in explicit neighbourhoods (\Cref{prop:splitting}).\\
Verified numerically &
The degree-three distribution identity and the full product over \(E_\rho[3]\); all shell data; the failure of the closed form at \(\tau=2\ii\).\\
Design proposal &
An \(8\times8\) phase-seamed microring torus provides a direct transmission-spectrum test (\Cref{sec:device}); square: a selector plus single-photon restriction gives a dual-rail qubit; triangular links: the threefold manifold gives a qutrit without a dimension-reducing selector.\\
Not claimed &
That the optimum detects complex multiplication; that the unit group has an established unique experimental signature (\Cref{rem:no-signature}); that the geometry fixes \(\tau\); a topologically protected degeneracy, a gate set, or a fabricated realization of the proposed twisted torus.\\
\bottomrule
\end{tabular}
\end{table}

The remainder of the paper is organized as follows.  \Cref{sec:spectra} formulates the twisted spectrum and separates gap optimization from arithmetic.  \Cref{sec:units} identifies the exceptional lattice units and derives the determinant response.  \Cref{sec:shells} gives the exact square and hexagonal shell data, and \Cref{sec:measurable} derives the measurable splitting and centroid laws.  \Cref{sec:qubit} explains the square-qubit and hexagonal-qutrit encodings.  \Cref{sec:device} then develops the quantitative microring experiment, including linewidth and calibration budgets.  The discussion and conclusions delimit what the construction does and does not establish; the appendices collect shell calculations and numerical checks.

\section{Twisted spectra and what selects the optimum}\label{sec:spectra}

\subsection{Setting}

Let \(\Gamma_\tau=\Z+\tau\Z\) be the period lattice and \(E_\tau=\C/\Gamma_\tau\) its complex torus, where \(\tau=x+\ii y\) and \(y>0\).  The Gram matrix of the period basis \((1,\tau)\) is \(G_\tau=\left(\begin{smallmatrix}1&x\\x&|\tau|^2\end{smallmatrix}\right)\), and \(G_\tau^{-1}\) is the corresponding dual metric on momentum coordinates.  Writing \(z=s+t\tau\) with \(s,t\in\R/\Z\), a flat unitary line bundle with holonomies \(e^{2\pi\ii u}\) around \(1\) and \(e^{2\pi\ii v}\) around \(\tau\) has sections \(f_{m,n}=\exp[2\pi\ii((m+u)s+(n+v)t)]\) and eigenvalues
\begin{equation}
 \lambda_{m,n}(a)=\kappa^2Q_\tau((m,n)+a),
 \qquad Q_\tau(q)=q^{T}G_\tau^{-1}q,
 \qquad \kappa=2\pi .
 \label{eq:torus-spectrum}
\end{equation}
Thus \(Q_\tau\) is the dimensionless dual-lattice quadratic form and \(\kappa\) is the spectral wavenumber scale (equal to \(2\pi\) in the unit-period normalization used here).  Set \(\delta_\tau(a)=\min_{k\in\Z^2}Q_\tau(k+a)\); maximizing \(\delta_\tau\) places \(a\) at a deep hole, a point realizing the covering radius \cite{ConwaySloane1999}.

\subsection{Gap maximization is a symmetry statement}

\begin{proposition}[Rectangular tori]\label{prop:rectangular}
For \(\tau=\ii t\), \(t>0\), \(\ \delta_\tau(a)=\dist(u,\Z)^2+\dist(v,\Z)^2/t^2\), whose unique maximizer is \(a=(1/2,1/2)\) for every \(t\).
\end{proposition}

\begin{proof}
\(G_\tau=\diag(1,t^2)\) is diagonal, so the minimization over \(k\in\Z^2\) separates, and each one-dimensional problem is uniquely maximized at the half-integer.
\end{proof}

\begin{proposition}[CM implies torsion optima]\label{prop:rational}
If \(\tau\) is imaginary quadratic then every maximizer of \(\delta_\tau\) is rational modulo \(\Z^2\), hence a torsion point of the holonomy torus.
\end{proposition}

\begin{proof}
From \(A\tau^2+B\tau+C=0\) with \(A,B,C\in\Z\), \(A\neq0\), we get \(\Re\tau=-B/2A\in\Q\) and \(|\tau|^2=C/A\in\Q\), so \(G_\tau\) and \(Q_\tau=G_\tau^{-1}\) are rational.  The function \(a\mapsto\dist_{Q_\tau}(a,\Z^2)\) is convex on each Voronoi cell, so its maximum is at a Voronoi vertex \cite{ConwaySloane1999}; such a vertex has at least three equidistant nearest lattice points, necessarily non-collinear.  Translating one to the origin and choosing linearly independent nearest vectors \(k_1,k_2\), equidistance gives \(2a^TQ_\tau k_j=k_j^TQ_\tau k_j\), a nonsingular rational system, so \(a\in\Q^2\).
\end{proof}

\begin{remark}[The theorem is correct; the physical reading was not]\label{rem:not-detectable}
\Cref{prop:rational} is a valid implication and we retain it.  What does not follow is the converse, nor any experimental reading.  Since \(G_\tau\) is rational up to scale exactly when \(\tau\) is imaginary quadratic, \Cref{prop:rational} is the statement ``rational form \(\Rightarrow\) rational deep holes''.  \Cref{prop:rectangular} shows the converse fails on a full one-parameter family: for \(\tau=\ii\pi\) and \(\tau=\ii\mathrm e\), neither imaginary quadratic, the optimum is exactly the \(2\)-torsion point \((1/2,1/2)\); we confirmed this on a \(600\times600\) grid, while the sheared modulus \(\tau=0.3+1.1\ii\) gives an optimum near \((0.413,0.623)\) with no evident rational structure.  Consequently, comparing a square device with an aspect-ratio-detuned ``non-CM'' device tests nothing: only a shear moves the optimum.
\end{remark}

\section{Units, symmetry-fixed holonomies, and the exact response}\label{sec:units}

\subsection{The unit group}

Let \(\Aut^{+}(Q_\tau)\subset SL_2(\Z)\) be the group of \emph{orientation-preserving} automorphisms of the period lattice \(\Gamma_\tau\), equivalently the units of its endomorphism ring acting by rotation.  (The full isometry group of a lattice is dihedral, containing reflections whenever the lattice is rectangular or rhombic; we use only the rotation subgroup.)  It is classical that \(\Aut^{+}(Q_\tau)\) is cyclic of order \(2\), \(4\) or \(6\), and that, up to homothety and \(SL_2(\Z)\) change of basis, order \(4\) occurs only for \(\tau=\ii\) and order \(6\) only for \(\tau=\rho=e^{\ii\pi/3}\), the curves with \(j=1728\) and \(j=0\) \cite{Silverman2009,ConwaySloane1999}.

\begin{theorem}[Symmetry-fixed holonomies]\label{thm:units}
Let \(S\in\Aut^{+}(Q_\tau)\) act on the holonomy torus by \(a\mapsto Sa\).
\begin{enumerate}[label=(\roman*)]
 \item If \(S\) has order \(N\ge3\), the fixed-point set is the group \((I-S)^{-1}\Z^2/\Z^2\), of order \(|\det(I-S)|\).
 \item Let
 \begin{equation*}
  S_4=\begin{pmatrix}0&-1\\1&0\end{pmatrix}\ \text{at }\tau=\ii,
  \qquad
  S_3=\begin{pmatrix}0&-1\\1&-1\end{pmatrix}\ \text{at }\tau=\rho .
 \end{equation*}
 Then \(\det(I-S_4)=2\) with fixed set \(\{(0,0),(1/2,1/2)\}\), and \(\det(I-S_3)=3\) with fixed set \(\{(0,0),(1/3,2/3),(2/3,1/3)\}\).  In both cases the nontrivial fixed points are exactly the deep holes of \(Q_\tau\).
 \item If \(F\) is a smooth \(S\)-invariant function with \(S\) of order \(N\ge3\), then \(\nabla F=0\) at every fixed point, and \(\Hess F\) there is proportional to \(G_\tau^{-1}\).
\end{enumerate}
\end{theorem}

\begin{proof}
(i)  \(a\) is fixed modulo \(\Z^2\) iff \((I-S)a\in\Z^2\).  Since \(S\) has order \(N\ge3\) its eigenvalues are primitive \(N\)th roots of unity, so \(1\) is not an eigenvalue and \(I-S\) is invertible; the fixed set is \((I-S)^{-1}\Z^2/\Z^2\), of order \(|\det(I-S)|\).

(iii)  Invariance gives \(S^{T}\nabla F=\nabla F\); as \(1\) is not an eigenvalue of \(S^{T}\), the gradient vanishes.  Differentiating twice, \(S^{T}(\Hess F)S=\Hess F\), so \(\Hess F\) is a symmetric form invariant under a cyclic group of order \(\ge3\) acting irreducibly on \(\R^2\).  By Schur's lemma the space of such forms is one-dimensional, and \(G_\tau^{-1}\) is one of them.
\end{proof}

\begin{remark}[The statement is delicate in two ways]\label{rem:delicate}
First, part (ii) must name the specific rotations.  At \(\tau=\rho\) the unit group also contains the order-six rotation \(S_6=\left(\begin{smallmatrix}1&-1\\1&0\end{smallmatrix}\right)\), for which \(\det(I-S_6)=1\): its only fixed point is the origin, and it \emph{exchanges} the two deep holes.  A statement quantified over an arbitrary element of order \(N\ge3\) would therefore be false.  The deep holes are fixed by \(S_3=S_6^{2}\), not by \(S_6\).

Second, part (iii) requires smoothness, and \(\delta_\tau\) does not have it: it is a minimum of finitely many quadratic branches and has a conical singularity at a deep hole, as \Cref{prop:splitting} makes explicit.  \(\delta_\tau\) is \(\Aut^{+}\)-invariant, and that invariance is what locates its maximizers, but no gradient or Hessian conclusion may be drawn from (iii) for \(\delta_\tau\).  Part (iii) is applied below only to \(\log D_\tau\), which is smooth away from the theta divisor.
\end{remark}

\subsection{Determinant, current and susceptibility}\label{sec:determinant}

For a nontrivial flat line bundle on \(E_\tau\), let \(D_\tau(a)\) denote the dimensionless zeta-regularized determinant of its twisted scalar Laplacian.  We fix its normalization by the second Kronecker limit formula \cite{Siegel1980,Lang1987}:
\begin{equation}
 D_\tau(u,v)=\exp\!\left[-\frac{2\pi(\Im\zeta)^2}{\Im\tau}\right]
 \left|\frac{\vartheta_1(\zeta\mid\tau)}{\eta(\tau)}\right|^2 ,
 \qquad \zeta=v-u\tau ,
 \label{eq:determinant}
\end{equation}
This is a zeta determinant in the Ray--Singer framework \cite{RaySinger1971} and the squared Quillen norm of a theta section \cite{Quillen1985}.  At rational holonomies, and only then in the usual arithmetic terminology, the theta quotient is represented by a Siegel function; suitable powers are modular units \cite{KubertLang1981}.

\begin{remark}[The cycle convention is forced]\label{rem:convention}
Under \(v\mapsto v+1\), \(\zeta\mapsto\zeta+1\) and \(\vartheta_1(\zeta+1)=-\vartheta_1(\zeta)\), so \(D_\tau\) is unchanged.  Under \(u\mapsto u+1\), \(\zeta\mapsto\zeta-\tau\), and the quasi-periodicity \(|\vartheta_1(\zeta+\tau)|=e^{\pi\Im\tau+2\pi\Im\zeta}|\vartheta_1(\zeta)|\) is cancelled exactly by the Gaussian prefactor.  Hence \eqref{eq:determinant} is well defined on the holonomy torus, and no other linear combination of \(u,v\) with \(\tau\) has this property.  Replacing \(\zeta\) by \(-\zeta\) or \(\bar\zeta\) leaves \(D_\tau\) invariant, which is the residual sign and cycle-exchange freedom.
\end{remark}

Define the determinant current \(\bm{\mathcal J}=-\nabla_a\log D_\tau\) and the determinant susceptibility (response tensor) \(\mathcal K=-\Hess_a\log D_\tau\), where both derivatives are taken in the holonomy coordinates \(a=(u,v)\).

\begin{theorem}[Exact response at symmetry-fixed holonomies]\label{thm:response}
Write \(L=\log D_\tau\) and \(\zeta=X+\ii Y\).  Then, away from the theta divisor,
\begin{equation}
 \Delta_\zeta L=-\frac{4\pi}{\Im\tau}
 \label{eq:laplacian}
\end{equation}
identically in \(\tau\) and \(a\).  Consequently, at \((\tau,a)=(\ii,(1/2,1/2))\) and at \((\rho,(1/3,2/3))\),
\begin{equation}
 \bm{\mathcal J}=0,
 \qquad
 \mathcal K=2\pi\,(\Im\tau)\,G_\tau^{-1}.
 \label{eq:unified-response}
\end{equation}
\end{theorem}

\begin{proof}
For \eqref{eq:laplacian}: \(L=-2\pi Y^2/\Im\tau+\log|\vartheta_1(\zeta)/\eta|^2\).  The second term is the real part of a holomorphic function of \(\zeta\) and is therefore harmonic away from the zeros of \(\vartheta_1\); the first contributes \(-(2\pi/\Im\tau)\,\partial_Y^2Y^2=-4\pi/\Im\tau\).

For \eqref{eq:unified-response}: the holonomies in question are fixed points of \(S_4\) and \(S_3\) respectively by \Cref{thm:units}(ii), and \(D_\tau\) is \(\Aut^{+}(Q_\tau)\)-invariant and smooth there, so \Cref{thm:units}(iii) gives \(\bm{\mathcal J}=0\) and \(\Hess_aL=-c\,G_\tau^{-1}\) for some \(c\).  The map \(a\mapsto(X,Y)\) is linear with Jacobian \(M=\left(\begin{smallmatrix}-x&1\\-y&0\end{smallmatrix}\right)\), and
\begin{equation}
 M^{T}M=\begin{pmatrix}|\tau|^2&-x\\-x&1\end{pmatrix}=y^{2}\,G_\tau^{-1},
\end{equation}
since \(\det G_\tau=y^2\).  From \(\Hess_aL=M^{T}(\Hess_\zeta L)M=-c\,G_\tau^{-1}=-(c/y^{2})M^{T}M\) and invertibility of \(M\) we get \(\Hess_\zeta L=-(c/y^{2})I_2\), whence \(\Delta_\zeta L=-2c/y^{2}\).  Comparing with \eqref{eq:laplacian} gives \(c=2\pi y=2\pi\Im\tau\).
\end{proof}

\begin{proposition}[Square determinant value]\label{prop:square-value}
\(D_\ii(1/2,1/2)=2\).
\end{proposition}

\begin{proof}
Here \(\zeta=(1-\tau)/2\equiv(1+\tau)/2\).  Using \(\vartheta_1(z+\tau/2)=\ii e^{-\ii\pi(\tau/4+z)}\vartheta_4(z)\) at \(z=1/2\), together with \(\vartheta_4(1/2)=\vartheta_3(0)\), gives \(\vartheta_1\bigl((1+\tau)/2\bigr)=e^{-\ii\pi\tau/4}\vartheta_3(0)\).  For \(\tau=\ii y\) the prefactor contributes \(|e^{-\ii\pi\tau/4}|^2=e^{\pi y/2}\), which cancels the Gaussian factor \(e^{-2\pi(y/2)^2/y}=e^{-\pi y/2}\).  Hence
\begin{equation}
 D_{\ii y}\left(\tfrac12,\tfrac12\right)=\frac{\vartheta_3(0\mid \ii y)^2}{\eta(\ii y)^2}.
\end{equation}
At \(y=1\) the lattice is self-dual, so \(\vartheta_2(0\mid\ii)=\vartheta_4(0\mid\ii)\).  With \(\vartheta_3^4=\vartheta_2^4+\vartheta_4^4\) this gives \(\vartheta_3^2=\sqrt2\,\vartheta_4^2\), and with \(\eta^3=\tfrac12\vartheta_2\vartheta_3\vartheta_4\) we obtain \(\eta^3=\tfrac12\vartheta_3\vartheta_4^2=\vartheta_3^3/(2\sqrt2)\), i.e.\ \(\eta(\ii)=\vartheta_3(0\mid\ii)/\sqrt2\).  Therefore \(D_\ii(1/2,1/2)=2\).
\end{proof}

\begin{proposition}[Hexagonal determinant value]\label{prop:hex-value}
At either hexagonal deep hole,
\begin{equation}
 D_\rho\!\left(\tfrac13,\tfrac23\right)
 =D_\rho\!\left(\tfrac23,\tfrac13\right)=\sqrt3 .
\end{equation}
\end{proposition}

\begin{proof}
With the normalization in \eqref{eq:determinant}, \(D_\tau(P)\) is the square
of the canonical Arakelov--Green function \(G_{E_\tau}(0,P)\).  Let
\(\phi=I-S_3:E_\rho\to E_\rho\) be the self-isogeny induced by the
hexagonal order-three unit.  By \Cref{thm:units}(ii),
\begin{equation}
 \deg\phi=|\det(I-S_3)|=3,
 \qquad
 \ker\phi=\left\{0,\left(\tfrac13,\tfrac23\right),
                         \left(\tfrac23,\tfrac13\right)\right\}.
\end{equation}
The energy formula for an isogeny \(f:E\to E'\) of degree \(N\) gives
\cite[Theorem~4.2]{deJong2005}
\begin{equation}
 \prod_{\substack{P\in\ker f\\P\ne0}}G_E(0,P)
 =\sqrt N\,\frac{\|\eta\|(E')^2}{\|\eta\|(E)^2}.
 \label{eq:isogeny-energy}
\end{equation}
Here \(E'=E=E_\rho\), so the eta quotient is one and the product equals
\(\sqrt3\).  The two nonzero kernel points are negatives of each other,
and the Green function is even and positive.  Hence
\(G_{E_\rho}(0,a_\rho)^2=\sqrt3\).  Since \(D_\rho=G_{E_\rho}^2\), this
is precisely \(D_\rho(a_\rho)=\sqrt3\).
\end{proof}

\begin{remark}[Distribution check on three-torsion]\label{rem:hex-value}
The same normalization yields
\begin{equation}
 \prod_{P\in E_\rho[3]\setminus\{0\}}D_\rho(P)=9,
\end{equation}
either by applying
the isogeny formula to multiplication by three or by the distribution relation
for Siegel units \cite{KubertLang1981}.  Numerically, the two deep holes have
value \(\sqrt3\) and the other six nonzero three-torsion points have value
\(3^{1/6}\), giving \(3(3^{1/6})^6=9\).  This full product is a consistency
check; the exact deep-hole value already follows from the degree-three
self-isogeny in \Cref{prop:hex-value}.
\end{remark}

\begin{remark}[The closed form fails without the extra unit]\label{rem:not-general}
Formula \eqref{eq:unified-response} does not extend to holonomies fixed only by \(a\mapsto-a\).  At \(\tau=2\ii\), \(a=(1/2,1/2)\),
\begin{equation}
 \Hess_a\log D=\diag(-23.9576,\,-0.29380),
 \qquad
 -2\pi\Im\tau\,G_\tau^{-1}=\diag(-12.5664,\,-3.14159).
\end{equation}
The order-two symmetry still forces \(\bm{\mathcal J}=0\), and \eqref{eq:laplacian} still holds --- the trace of \(\Hess_\zeta L\) is \(-2\pi=-4\pi/\Im\tau\) --- but nothing forces \(\Hess_\zeta L\) to be isotropic, and it is not.  Isotropy of the response is thus a genuine consequence of the extra unit, available only at \(j=1728\) and \(j=0\).
\end{remark}

\section{Exact spectral data}\label{sec:shells}

\subsection{Shell structure at the two exceptional moduli}

The exact square and hexagonal data are compared in \Cref{tab:compare}; subsequent propositions justify the determinant entries and the measurable consequences of the shell multiplicities.

\begin{table}[ht]
\centering
\caption{The two exceptional moduli compared.  \(c\) is the coordination number of the deep hole, which fixes the dimension of the lowest manifold.}
\label{tab:compare}
\begin{tabular}{lll}
\toprule
& \(\tau=\ii\) (square, \(j=1728\)) & \(\tau=\rho\) (hexagonal, \(j=0\)) \\
\midrule
Unit group \(\Aut^{+}(Q_\tau)\) & cyclic of order \(4\) & cyclic of order \(6\)\\
Relevant rotation & \(S_4\), quarter turn & \(S_3=S_6^2\), \(120^\circ\)\\
Deep hole(s) & \((1/2,1/2)\) & \((1/3,2/3),\ (2/3,1/3)\)\\
\(\delta_\tau\) & \(1/2\) & \(4/9\)\\
Coordination \(c\) & \(4\) & \(3\)\\
Next shell & \(5/2\) (multiplicity \(8\)) & \(16/9\) (multiplicity \(3\))\\
Shell gap & \(2\kappa^2\) & \(\tfrac43\kappa^2\)\\
Determinant \(D_\tau\) & \(2\) (\Cref{prop:square-value}) & \(\sqrt3\) (\Cref{prop:hex-value})\\
Susceptibility \(\mathcal K\) & \(2\pi I_2\) & \(\dfrac{2\pi}{\sqrt3}\left(\begin{smallmatrix}2&-1\\-1&2\end{smallmatrix}\right)\)\\
Natural encoding & dual-rail qubit after a selector & qutrit, no selector needed\\
\bottomrule
\end{tabular}
\end{table}

At \(\tau=\ii\), \(a=(1/2,1/2)\), the momenta \((m,n)\in\{-1,0\}^2\) give \(Q=1/2\); the next eight give \(Q=5/2\).  At \(\tau=\rho\), \(Q_\rho(q)=\tfrac43(q_1^2-q_1q_2+q_2^2)\) and \(a=(1/3,2/3)\), the momenta \((0,0),(0,-1),(-1,-1)\) give \(Q=4/9\) and the next three give \(16/9\).  In both cases the offsets \(k_i+a\) from the deep hole to its nearest lattice points sum to zero, which is the content of the rotational symmetry and is used in \Cref{prop:splitting}.

\subsection{Ideal round benchmark (square case only)}

The round three-sphere in Hopf coordinates, \(\dd s^2=R^2(\dd\chi^2+\cos^2\!\chi\,\dd\phi^2+\sin^2\!\chi\,\dd\psi^2)\), with a scalar field \(\propto e^{\ii(m+u)\phi+\ii(n+v)\psi}\), reduces to a Jacobi-type problem whose regular solution behaves as \((\cos\chi)^{|m+u|}(\sin\chi)^{|n+v|}\).  In the zero-core Friedrichs realization,
\begin{equation}
 \lambda_{j,m,n}(a)=\frac{\ell(\ell+2)}{R^2}+\mu^2,
 \qquad \ell=2j+|m+u|+|n+v|,\quad j\ge0 .
 \label{eq:round-spectrum}
\end{equation}
For \(a=0\) this reproduces the degeneracies \((\ell+1)^2\), which we verified for \(\ell\le3\).

\begin{proposition}\label{prop:round}
The lowest eigenvalue of \eqref{eq:round-spectrum} is maximized uniquely at \(a=(1/2,1/2)\), where for \(\mu=0\) the two lowest distinct shells are \(3/R^2\) (fourfold) and \(8/R^2\), with separation \(5/R^2\).
\end{proposition}

\begin{proof}
\(\ell_0(a)=\dist(u,\Z)+\dist(v,\Z)\le1\) with equality only at the antiperiodic point, and \(\ell\mapsto\ell(\ell+2)\) is increasing.
\end{proof}

This model minimizes an \(\ell^1\) quantity, not a quadratic form, so neither \Cref{prop:rational} nor \Cref{thm:units} is used: the optimum is the elementary maximization of \(\dist(u,\Z)+\dist(v,\Z)\).  It has no hexagonal analogue.  Its gaps are properties of the Friedrichs extension at the two singular axes; other self-adjoint extensions exist, and a finite core interpolates, so the numbers are idealized benchmarks.

\section{Measurable predictions}\label{sec:measurable}

The determinant of \Cref{sec:determinant} is a collective one-loop quantity and is not a laboratory observable; we therefore do not propose it as a test.  The measurable content is the resonance-frequency pattern.

\begin{proposition}[Linear splitting, stationary centroid]\label{prop:splitting}
Let \(a_0\) be a deep hole of \(Q_\tau\) with nearest offsets \(w_i=k_i+a_0\), \(i=1,\dots,c\), and let \(a=a_0+\varepsilon\).  Then, \emph{exactly},
\begin{equation}
 Q_\tau(w_i+\varepsilon)=\delta_\tau+2\,w_i^{T}G_\tau^{-1}\varepsilon+Q_\tau(\varepsilon),
\end{equation}
so the manifold splits linearly.  If, in addition, \(\sum_iw_i=0\), as at the square and hexagonal symmetry-fixed deep holes, the centroid obeys
\begin{equation}
 \overline{Q}=\delta_\tau+Q_\tau(\varepsilon),
 \qquad
 \Hess\,\overline{Q}=2\,G_\tau^{-1} .
\end{equation}
These branches are the \(c\) lowest eigenvalues provided
\(\|\varepsilon\|_{Q_\tau}<\tfrac12\bigl(\sqrt{\delta^{(2)}}-\sqrt{\delta_\tau}\bigr)\),
where \(\delta^{(2)}\) is the next shell value; for \(\tau=\ii\) the sharp condition is
\begin{equation}
 2|\varepsilon_u|+|\varepsilon_v|<1
 \quad\text{and}\quad
 |\varepsilon_u|+2|\varepsilon_v|<1 ,
 \label{eq:sharp-domain}
\end{equation}
for which \(Q=\tfrac12\pm\varepsilon_u\pm\varepsilon_v+\varepsilon_u^2+\varepsilon_v^2\) with independent signs.  In the round model \eqref{eq:round-spectrum} the four lowest eigenvalues are \((1+s)(3+s)/R^2\) with \(s=\pm\varepsilon_u\pm\varepsilon_v\), and \(\overline\lambda=(3+\varepsilon_u^2+\varepsilon_v^2)/R^2\), valid precisely for \(|\varepsilon_u|+|\varepsilon_v|<1/2\).
\end{proposition}

\begin{proof}
The first identity is the exact expansion of a quadratic form and does not require the centroid hypothesis.  For \(\tau=\ii\) the offsets are \((\pm\tfrac12,\pm\tfrac12)\) and for \(\tau=\rho\) they are \((\tfrac13,\tfrac23),(\tfrac13,-\tfrac13),(-\tfrac23,-\tfrac13)\); in both cases \(\sum_iw_i=0\), which kills the linear term in the mean.  A generic deep hole need not satisfy this additional identity.  For the validity domain, \(Q_\tau^{1/2}\) is a norm, so each \(\sqrt{Q_\tau(k+a_0+\varepsilon)}\) changes by at most \(\|\varepsilon\|_{Q_\tau}\); the stated inequality keeps the first shell below the second.  The sharp square condition compares the largest branch \(\tfrac12+|\varepsilon_u|+|\varepsilon_v|+\varepsilon_u^2+\varepsilon_v^2\) with the lowest intruder \(\tfrac52-3|\varepsilon_u|-|\varepsilon_v|+\varepsilon_u^2+\varepsilon_v^2\) and its transpose.  In the round model the largest branch is \(1+|\varepsilon_u|+|\varepsilon_v|\) and the lowest intruder has \(\ell=2-|\varepsilon_u|-|\varepsilon_v|\).
\end{proof}

We verified both the identities and the sharpness of the domains numerically; the condition \eqref{eq:sharp-domain} is necessary as well as sufficient, and at \(\varepsilon=(0.34,0.34)\), for instance, the fourth branch is no longer among the four lowest.  \Cref{fig:splitting} illustrates, in the order square then hexagonal, the linear branch splitting, the next-shell gap, and the stationary centroid predicted by \Cref{prop:splitting}.

\begin{figure}[H]
\centering
\begin{tikzpicture}[x=1cm,y=1cm,>=Latex,font=\small,
 axis/.style={->,line width=0.7pt,mutedgray},
 lev/.style={line width=1.5pt,protectblue},
 cent/.style={line width=1.0pt,protectgreen,dashed},
 higher/.style={line width=1.3pt,protectred}]

% ---- panel (a): square, fourfold ----
\begin{scope}
\node[font=\bfseries,align=center] at (3.0,4.30) {(a) $\tau=\ii$: fourfold $\to$ qubit};
\draw[axis] (0,0.15)--(0,3.95);
\draw[axis] (0,1.95)--(6.1,1.95);
\node[font=\scriptsize,rotate=90,text=mutedgray] at (-0.42,2.4) {spectral value};
\node[font=\scriptsize,text=mutedgray] at (3.0,-0.30) {phase error along a generic ray};
\draw[higher] (0.3,3.60)--(5.8,3.60);
\node[protectred,anchor=west,font=\scriptsize] at (4.05,3.83) {next shell};
\draw[lev] (0.3,0.97)--(3.05,1.95)--(5.8,0.97);
\draw[lev] (0.3,1.53)--(3.05,1.95)--(5.8,1.53);
\draw[lev] (0.3,2.37)--(3.05,1.95)--(5.8,2.37);
\draw[lev] (0.3,2.93)--(3.05,1.95)--(5.8,2.93);
\draw[cent] (0.3,2.12)--(3.05,1.95)--(5.8,2.12);
\filldraw[protectblue] (3.05,1.95) circle (2.0pt);
\node[protectgreen,anchor=west,font=\scriptsize] at (3.25,1.62) {centroid};
\node[mutedgray,anchor=west,font=\scriptsize] at (0.28,0.55) {$c=4$, gap $2\kappa^2$};
\end{scope}

% ---- panel (b): hexagonal, threefold ----
\begin{scope}[xshift=7.6cm]
\node[font=\bfseries,align=center] at (3.0,4.30) {(b) $\tau=\rho$: threefold $\to$ qutrit};
\draw[axis] (0,0.15)--(0,3.95);
\draw[axis] (0,1.95)--(6.1,1.95);
\node[font=\scriptsize,text=mutedgray] at (3.0,-0.30) {phase error along a generic ray};
\draw[higher] (0.3,3.30)--(5.8,3.30);
\node[protectred,anchor=west,font=\scriptsize] at (4.05,3.53) {next shell};
\draw[lev] (0.3,1.10)--(3.05,1.95)--(5.8,1.10);
\draw[lev] (0.3,1.80)--(3.05,1.95)--(5.8,1.80);
\draw[lev] (0.3,2.85)--(3.05,1.95)--(5.8,2.85);
\draw[cent] (0.3,2.10)--(3.05,1.95)--(5.8,2.10);
\filldraw[protectblue] (3.05,1.95) circle (2.0pt);
\node[mutedgray,anchor=west,font=\scriptsize] at (0.28,0.55) {$c=3$, gap $\tfrac43\kappa^2$};
\end{scope}
\end{tikzpicture}
\caption{Behaviour predicted by \Cref{prop:splitting}, schematic.  At a deep hole the lowest manifold splits \emph{linearly} in the phase error, with slopes \(2w_i^{T}G_\tau^{-1}\varepsilon\); the lowest level therefore has a cusp rather than a smooth maximum.  Because \(\sum_iw_i=0\), the centroid (dashed) is stationary with Hessian \(2G_\tau^{-1}\).  The square lattice gives a fourfold manifold requiring a selector to reach a qubit; the hexagonal lattice gives a threefold manifold matching a qutrit directly.  Slopes are drawn schematically; the exact square values are \(\pm\varepsilon_u\pm\varepsilon_v\).}
\label{fig:splitting}
\end{figure}
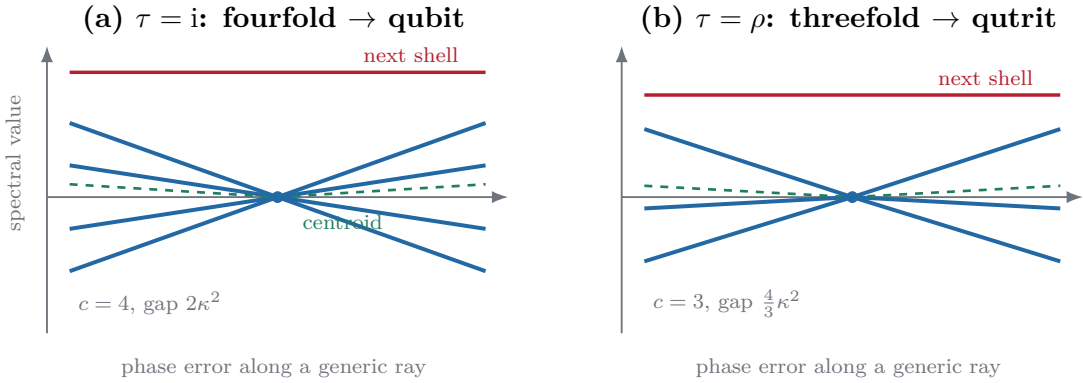

\Cref{prop:splitting} is the operational content: it predicts a \emph{cusp} rather than a smooth maximum in the lowest resonance as a function of phase, with unit slopes for coordinate-axis scans in the square case, together with a stationary shell centroid of curvature \(2G_\tau^{-1}\).  Both are visible in a transmission spectrum.

\begin{remark}[What the centroid does \emph{not} test]\label{rem:no-signature}
The centroid curvature \(2G_\tau^{-1}\) follows from the quadratic dispersion \eqref{eq:torus-spectrum} alone; it holds at the deep hole of \emph{any} lattice whose nearest offsets sum to zero, including every rectangular one.  It therefore does not detect the unit group.  We are not aware of a directly measurable quantity whose value distinguishes \(\tau=\ii\) or \(\tau=\rho\) from a nearby non-exceptional modulus: the quantity that the extra unit does control, \eqref{eq:unified-response}, is a determinant response and is not observable.  Accordingly, \Cref{thm:units,thm:response} should be read as structural results that explain and unify the exact response data, not as an experimental signature of arithmetic.  We flag this explicitly because the opposite reading is the natural temptation, and it would not be supported.
\end{remark}

\section{Encodings: square qubit, hexagonal qutrit}\label{sec:qubit}

\subsection{Quantization is the step that creates a logical system}

Let \(\psi_0,\dots,\psi_{d-1}\) be orthogonal modes selected from the lowest manifold, with bosonic operators \(a_i\).  The single-excitation subspace
\begin{equation}
 \cH_Q=\operatorname{span}\bigl\{|i_L\rangle=a_i^\dagger|\mathrm{vac}\rangle\bigr\}\cong\C^{d}
 \label{eq:dual-rail}
\end{equation}
is a \(d\)-rail logical system: the two-rail case is the standard photonic qubit \cite{KokEtAl2007}, while multimode quantum optics and path-encoded qutrit experiments support the general modal interpretation \cite{FabreTreps2020,LuEtAl2020}.  The classical field \(\sum_ic_i\psi_i\) has the same mode geometry and is a useful coherent simulator, but it is a qubit or qutrit only after restriction to the single-photon sector.

\subsection{Square: fourfold parent manifold, qubit after a selector}

At \(\tau=\ii\) the lowest manifold has \(c=4\), so a device-specific selector \(V\) --- polarization, mirror parity, helicity, or a designed anisotropy --- must isolate a logical pair and detune the remaining two.  With \(P\) the projector onto the chosen pair and \(Q=I-P\), the leakage gap \(\Delta_{\rm leak}=\inf\operatorname{spec}(QHQ)-\sup\operatorname{spec}(PHP)\) is controlled by the parent-shell scales \(2\kappa^2\) and \(5/R^2\), and if \(\|V\|\ll\Delta_{\rm parent}\) subspace perturbation theory bounds higher-shell admixture by \(O(\|V\|/\Delta_{\rm parent})\).

\subsection{Hexagonal: threefold manifold, qutrit without a dimension-reducing selector}

At \(\tau=\rho\) the lowest manifold has \(c=3\) and matches the dimension of a qutrit exactly, with shell gap \(\tfrac43\kappa^2\).  No dimension-reducing selector is required, which removes the main structural awkwardness of the square case; mode-resolved preparation, control, and readout are still required.  The trade-off is a smaller relative gap and the absence of an exactly solvable round benchmark.  Whether a qutrit or a selector-reduced qubit is preferable is a device question, but the hexagonal geometry is the cleaner realization of the mode-counting principle.

\subsection{What the gap does not do}

Three qualifications, stated without hedging.  First, this is spectral isolation, not topological protection: the degeneracy is a symmetry degeneracy, and \Cref{prop:splitting} shows that a generic phase error lifts it at \emph{first} order.  No topological invariant, anyonic structure, or code space is claimed.  Second, the mechanism is of the same type as a large free spectral range; the interest lies in the exact accompanying data, not in the existence of a gap.  Third, the holonomies commute, so within \(\cH_Q\) they generate phase gates only; universal control requires an additional noncommuting mode coupler, and entangling operations require further optical resources.  Finally, stationarity of the determinant and of the shell centroid concerns collective quantities; it does not imply that a \emph{logical} splitting is first-order insensitive, which is a property of the selector \(V\) and must be established separately.

\begin{remark}[Relation to the character-variety picture]\label{rem:fricke}
A commuting \(\SU(2)\) pair is simultaneously diagonalizable, and the trace coordinates \(x=2\cos2\pi u\), \(y=2\cos2\pi v\), \(z=2\cos2\pi(u+v)\) satisfy the Cayley cubic \(x^2+y^2+z^2-xyz-4=0\), one of the character surfaces studied in \cite{PlanatEtAl2022Character}.  Its compact \(\SU(2)\) locus is the pillowcase used in the Fricke topological-qubit proposal \cite{PlanatEtAl2022Fricke}; the antiperiodic point maps to the node \((-2,-2,2)\) with \(U_a=U_b=-I\).  We record this but do not build on it: the spectrum \eqref{eq:torus-spectrum} depends only on the abelian eigenphases, the internal doublet plays no role anywhere in this paper, and at the node the holonomies are central.  Everything here holds verbatim for \(\U(1)^2\).  Recovering a role for the nonabelian structure requires a link with nonabelian group, a different problem.
\end{remark}

The compact locus and the relevant central character are shown in \Cref{fig:cayley-node}; the figure is a geometric pointer to the prior Fricke construction, not evidence for the present spectral claims.
\begin{figure}[H]
\centering
\includegraphics[width=0.66\textwidth]{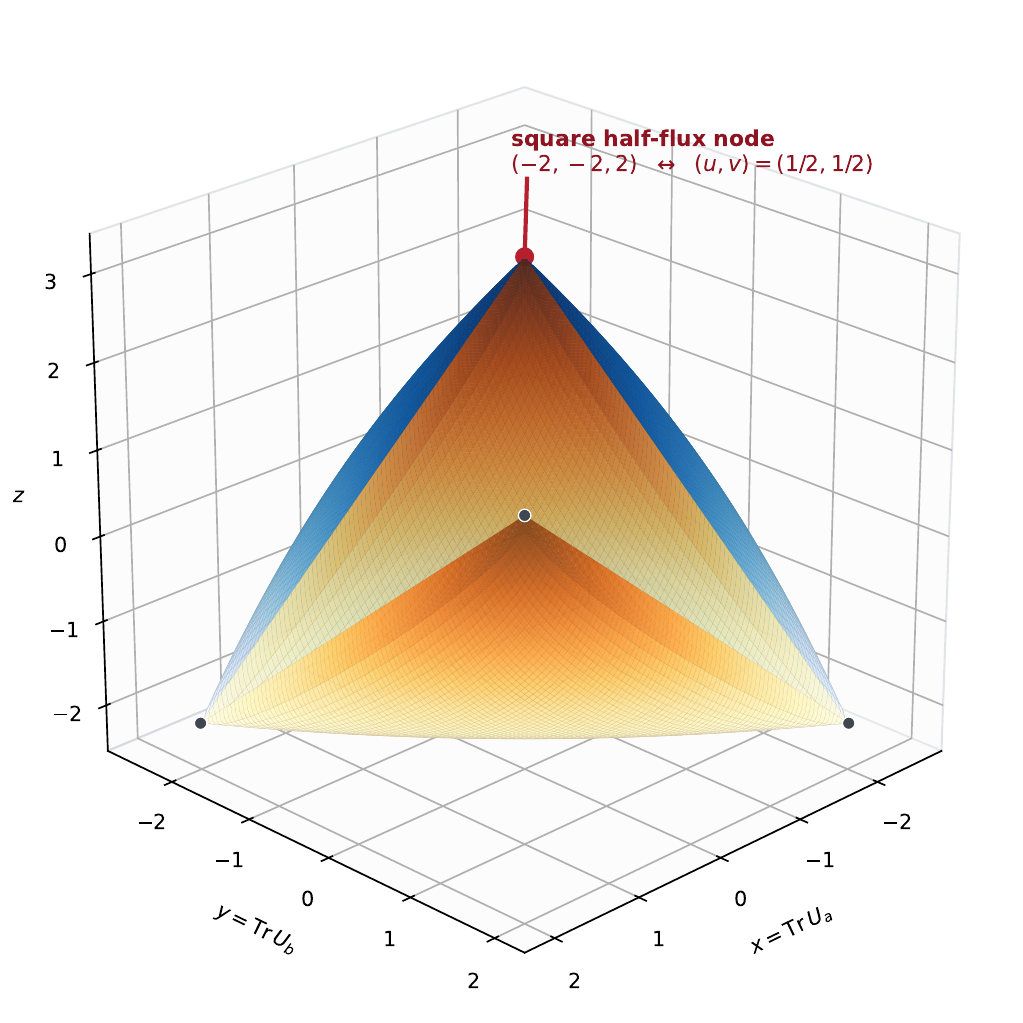}
\caption{The compact \(\SU(2)\) locus of the Cayley cubic, with sheets \(z_\pm=\tfrac12(xy\pm\sqrt{(x^2-4)(y^2-4)})\) meeting at four central characters; the highlighted node \((-2,-2,2)\) corresponds to \((u,v)=(1/2,1/2)\).  As \Cref{rem:fricke} explains, this is a labelling of the holonomy torus and carries none of the spectral content.}
\label{fig:cayley-node}
\end{figure}

\section{A quantitative integrated-photonic test}\label{sec:device}

\subsection{Device selected: a phase-seamed microring torus}

We select one specific platform: an \(N\times N\) array of nominally identical silicon site microrings, coupled through off-resonant link rings.  Arrays as large as \(10\times10\), with a reported tunnelling scale near \(16\) GHz, have already been fabricated and spectrally imaged \cite{HafeziEtAl2013}.  The proposed first device uses \(N=8\), hence 64 site rings.  All bulk links are real and equal.  The right-to-left boundary links carry the common hopping phase \(2\pi u\), and the top-to-bottom links carry \(2\pi v\).  The phase around every elementary plaquette is zero: only the two noncontractible products are nontrivial.  All links within each seam family are tuned together by thermo-optic or electro-optic phase shifters.

The two wrap-around seams are the principal fabrication risk.  There are \(2N=16\) boundary links; a naive planar routing in which each passes as many as \(N-1=7\) bulk rows produces up to \(112\) crossings, although multilayer routing or perimeter paths could reduce this layout-dependent number.  More importantly, an off-resonant link ring generally controls both the hopping amplitude and phase, so the long seam links cannot simply be assumed equivalent to bulk links.  Their hopping rates must meet the tolerance in \Cref{tab:tolerances}.  Modulated-ring frequency lattices provide demonstrated phase-programmable long-range hopping and band spectroscopy \cite{DuttEtAl2020,DinhEtAl2024}, but a frequency-mode chain is not automatically a finite periodic two-torus: a second periodic coordinate and controlled boundary identification would still have to be engineered.  We therefore retain the spatial array as the primary proposal because it realizes the required \(N\times N\) graph and permits site-resolved imaging; synthetic dimensions are a promising alternative architecture, not a seamless realization already supplied by the cited experiments.

Weak input and drop buses couple to several non-equivalent site rings; alternatively, the out-of-plane scattered field is imaged site by site as in Ref.~\cite{HafeziEtAl2013}.  A narrow-linewidth C-band laser is swept over a \(30\)--\(40\) GHz window and the complex multiport transmission is recorded by an optical vector analyser.  Multiple readout channels matter: precisely at a degeneracy a scalar transmission trace contains one pole, not a visible count of coincident modes.  The multiplicity is obtained from the rank of the fitted pole residue or, more simply, by counting the branches after a generic two-coordinate detuning.  The probe coupling must be weak: unlike the strongly overcoupled imaging experiment of Ref.~\cite{HafeziEtAl2013}, the objective here is to preserve linewidth.  Weak coupling reduces extraction efficiency, but the reduction depends on the complete add/drop geometry and cannot be inferred from the coupling-rate ratio alone; it must be included in an experimental signal-to-noise budget.  Coherent vector detection is assumed because it fits complex poles even when intensity maxima overlap.  Adding one diagonal link per cell changes the square graph into a triangular graph and tests the hexagonal qutrit case on the same principle.

Each site ring also has clockwise and counterclockwise circulation modes.  Reciprocal link phases give them conjugate holonomies, so exciting both would double every multiplicity.  Directional bus coupling, as used in ring-array experiments, selects one circulation.  Residual backscattering between the two circulations is included as a calibration error below; this selection is distinct from the later dimension-reducing selector required to turn the square fourfold manifold into a qubit.

The proposed square topology and the multiport readout chain are summarized in \Cref{fig:device}, which follows the two conceptual figures and precedes the finite-device formulas.

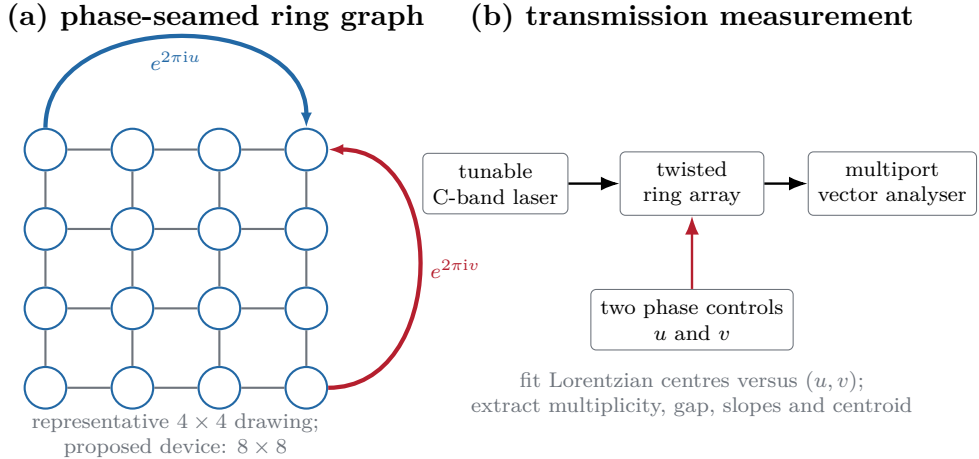
\begin{figure}[H]
\centering
\begin{tikzpicture}[font=\small,>=Latex,
 site/.style={circle,draw=protectblue,fill=white,minimum size=5.5mm,inner sep=0pt,line width=0.8pt},
 bulk/.style={mutedgray,line width=0.8pt},
 useam/.style={protectblue,line width=1.6pt,-{Latex[length=2mm]}},
 vseam/.style={protectred,line width=1.6pt,-{Latex[length=2mm]}},
 box/.style={draw=mutedgray,rounded corners=2pt,align=center,minimum height=8mm,minimum width=19mm,fill=white,font=\scriptsize}]
% representative lattice
\node[font=\bfseries] at (2.25,4.90) {(a) phase-seamed ring graph};
\foreach \x in {0,...,3}{
  \foreach \y in {0,...,3}{\node[site] (n\x\y) at (1.15*\x,1.05*\y) {};}}
\foreach \y in {0,...,3}{\foreach \x in {0,...,2}{\pgfmathtruncatemacro{\xp}{\x+1}\draw[bulk] (n\x\y)--(n\xp\y);}}
\foreach \x in {0,...,3}{\foreach \y in {0,...,2}{\pgfmathtruncatemacro{\yp}{\y+1}\draw[bulk] (n\x\y)--(n\x\yp);}}
\draw[useam] (n03.north) to[out=90,in=90,looseness=1.18] node[below=1pt,font=\scriptsize] {$e^{2\pi\ii u}$} (n33.north);
\draw[vseam] (n30.east) to[out=0,in=0,looseness=1.30] node[right,font=\scriptsize] {$e^{2\pi\ii v}$} (n33.east);
\node[font=\scriptsize,text=mutedgray,align=center] at (1.72,-0.63) {representative $4\times4$ drawing;\\proposed device: $8\times8$};
% readout chain
\node[font=\bfseries] at (8.55,4.90) {(b) transmission measurement};
\node[box] (laser) at (5.95,2.70) {tunable\\C-band laser};
\node[box] (chip) at (8.55,2.70) {twisted\\ring array};
\node[box,minimum width=22mm] (det) at (11.20,2.70) {multiport\\vector analyser};
\draw[->,line width=0.9pt] (laser)--(chip);
\draw[->,line width=0.9pt] (chip)--(det);
\node[box,minimum width=24mm] (ctrl) at (8.55,0.90) {two phase controls\\$u$ and $v$};
\draw[->,protectred,line width=0.9pt] (ctrl)--(chip);
\node[align=center,font=\scriptsize,text=mutedgray] at (8.55,-0.08) {fit Lorentzian centres versus $(u,v)$;\\extract multiplicity, gap, slopes and centroid};
\end{tikzpicture}
\caption{Concrete spectroscopy proposal.  The graph drawing suppresses the physical link rings and shows only a representative \(4\times4\) array.  In the \(8\times8\) device, two phase-controlled seam families close the square lattice without plaquette flux.  Weak multiport coupling and a swept laser measure both pole positions and residue rank.  Diagonal links give the triangular configuration used for the qutrit test.}
\label{fig:device}
\end{figure}

\subsection{Exact finite-device spectrum}

Write all optical frequencies and rates in cycles per second and let \(j\) denote the calibrated nearest-neighbour hopping rate.  In the one-photon or linear classical sector the square device is the graph Laplacian
\begin{equation}
 \frac{H_\square}{h}=\nu_{\rm ref}I+jL_\square(u,v),
 \qquad
 b_{r+N,s}=e^{2\pi\ii u}b_{r,s},\quad
 b_{r,s+N}=e^{2\pi\ii v}b_{r,s}.
 \label{eq:device-hamiltonian}
\end{equation}
Its \(N^2\) resonance centres are exactly
\begin{equation}
 \nu^{\square}_{pq}(u,v)=\nu_{\rm ref}+4j
 -2j\!\left[
 \cos\frac{2\pi(p+u)}{N}+\cos\frac{2\pi(q+v)}{N}
 \right],
 \qquad p,q\in\Z/N\Z .
 \label{eq:square-device-spectrum}
\end{equation}
Thus the finite device is not assumed to have the continuum dispersion \eqref{eq:torus-spectrum}: it has the measured cosine dispersion \eqref{eq:square-device-spectrum}, which approaches it near the band bottom with scale \(j(2\pi/N)^2\).

The finite device retains the relevant rotational equivariance, not merely the continuum dispersion.  The quarter turn acts on the holonomy torus by \((u,v)\mapsto(v,-u)\), fixing \((1/2,1/2)\) modulo \(\Z^2\), while the \(120^\circ\) rotation fixes the two hexagonal deep holes.  These operations are automorphisms of the corresponding finite graphs with their seam phases and therefore make the degeneracies below exact.  Symmetry alone does not prove global optimality for a nonsmooth lowest-eigenvalue function; for the square graph the maximum follows directly by separating the two cosine terms, and for the selected triangular device it is confirmed by exhaustive evaluation of the exact \(N=8\) spectrum.  Thus no continuum-limit argument is used for the finite-device claims.

At \((u,v)=(1/2,1/2)\) the four lowest modes are \((p,q)\in\{0,-1\}^2\), with
\begin{align}
 \nu_\star^\square-\nu_{\rm ref}
 &=4j\left(1-\cos\frac{\pi}{N}\right),\\
 \Delta_N^\square
 &=2j\left(\cos\frac{\pi}{N}-\cos\frac{3\pi}{N}\right).
 \label{eq:square-device-gap}
\end{align}
For \(u=1/2+\varepsilon_u\), \(v=1/2+\varepsilon_v\), their four exact branches are
\begin{equation}
 \nu_{st}^\square=\nu_{\rm ref}+4j-2j\!\left[
 \cos\frac{2\pi(\varepsilon_u+s/2)}{N}
 +\cos\frac{2\pi(\varepsilon_v+t/2)}{N}\right],
 \qquad s,t=\pm1 .
 \label{eq:square-device-branches}
\end{equation}
To second order,
\begin{equation}
 \nu_{st}^\square=\nu_\star^\square
 +A_N(s\varepsilon_u+t\varepsilon_v)
 +\frac{C_N}{2}(\varepsilon_u^2+\varepsilon_v^2)+O(\|\varepsilon\|^3),
 \label{eq:device-expansion}
\end{equation}
where
\begin{equation}
 A_N=\frac{4\pi j}{N}\sin\frac{\pi}{N},
 \qquad
 C_N=\frac{8\pi^2j}{N^2}\cos\frac{\pi}{N}.
\end{equation}
The lowest fitted line therefore has the cusp
\(\nu_{\min}=\nu_\star^\square-A_N(|\varepsilon_u|+|\varepsilon_v|)+O(\|\varepsilon\|^2)\).  The four-line centroid is exactly
\begin{equation}
 \bar\nu_\square=\nu_{\rm ref}+4j-2j\cos\frac{\pi}{N}
 \left[\cos\frac{2\pi\varepsilon_u}{N}+\cos\frac{2\pi\varepsilon_v}{N}\right],
 \label{eq:square-device-centroid}
\end{equation}
so its coordinate curvature at the optimum is \(C_N\).

For equal diagonal and axial couplings, the triangular configuration has
\begin{equation}
 \nu^{\triangle}_{pq}=\nu_{\rm ref}+6j-2j\left[
 \cos\theta_p+\cos\theta_q+\cos(\theta_p-\theta_q)\right],
 \quad \theta_p=\frac{2\pi(p+u)}{N},\quad
 \theta_q=\frac{2\pi(q+v)}{N}.
 \label{eq:triangular-device-spectrum}
\end{equation}
Realizing \eqref{eq:triangular-device-spectrum} requires three seam settings rather than two: with the diagonal along \(a_1-a_2\), a diagonal link crossing the right seam carries \(e^{2\pi\ii u}\), one crossing the top seam carries \(e^{-2\pi\ii v}\), and the corner link carries \(e^{2\pi\ii(u-v)}\).  These are determined by \(u\) and \(v\) but must be implemented as separate phase families, in addition to matching the diagonal hopping rate to the axial one.

At \((u,v)=(1/3,2/3)\), the modes \((0,0),(0,-1),(-1,-1)\) are exactly degenerate.  With \(\alpha=2\pi/N\),
\begin{align}
 \nu_\star^\triangle-\nu_{\rm ref}
 &=6j-2j\left(2\cos\frac{\alpha}{3}+\cos\frac{2\alpha}{3}\right),\\
 \Delta_N^\triangle
 &=2j\left(2\cos\frac{\alpha}{3}-\cos\frac{2\alpha}{3}
 -\cos\frac{4\alpha}{3}\right).
 \label{eq:triangular-device-gap}
\end{align}
For a scan in \(u\), their first-order slopes are \(-B_N,0,B_N\), where
\begin{equation}
 B_N=\frac{4\pi j}{N}\left[
 \sin\frac{2\pi}{3N}+\sin\frac{4\pi}{3N}\right].
 \label{eq:triangular-device-slope}
\end{equation}

For direct comparison with a transmission measurement, \Cref{tab:device-numbers} lists the square and triangular line positions and shell gaps at the selected device scale.

\begin{table}[t]
\centering
\caption{Directly measurable predictions for \(N=8\) and \(j=16\) GHz.  Frequencies are offsets from \(\nu_{\rm ref}\).  The final column is the exact spectrum at \(\varepsilon_u=0.10\), \(\varepsilon_v=0\), not a Taylor approximation.}
\label{tab:device-numbers}
\small
\begin{tabular}{@{}>{\raggedright\arraybackslash}p{0.19\textwidth}c c c >{\raggedright\arraybackslash}p{0.25\textwidth}@{}}
\toprule
Graph and point & Multiplicity & \(\nu_\star-\nu_{\rm ref}\) & Shell gap & Lines at \(\varepsilon_u=0.10\) \\
\midrule
Square, \((1/2,1/2)\) & 4 & 4.872 GHz & 17.318 GHz & 4.002 (twice), 5.924 (twice) GHz\\
Triangular, \((1/3,2/3)\) & 3 & 6.468 GHz & 18.106 GHz & 4.743, 6.659, 8.554 GHz\\
\bottomrule
\end{tabular}
\end{table}

The square doublet separation in the last column is \(1.922\) GHz; the two adjacent triangular separations are \(1.915\) and \(1.895\) GHz.  A small generic square detuning \((\varepsilon_u,\varepsilon_v)=(0.15,0.075)\) produces four distinct poles at \(2.968,4.409,5.846,7.288\) GHz, with a minimum adjacent separation of \(1.437\) GHz.  Distinct poles do not necessarily produce four intensity maxima.

We therefore separate visual peak resolution from complex pole estimation.  For the sum of \emph{two isolated, equal, incoherent} Lorentzians of full width at half maximum \(\Gamma\) and separation \(\Delta\), the midpoint becomes a minimum exactly when \(\Gamma<\sqrt3\,\Delta\).  This gives a two-line merging threshold \(3.33\) GHz for the \(\varepsilon_u=0.10\) square doublet.  It must not be applied pairwise to a four-line cluster: neighbouring tails interact, and an equal-residue simulation of the small generic quartet has only two intensity maxima at \(\Gamma=2\) GHz.  The four complex poles remain identifiable in a calibrated multiport fit, but the manuscript does not count them as four visually resolved peaks.

For the linewidth, Ref.~\cite{HafeziEtAl2013} quotes \(\kappa_{\rm in}\approx1\) GHz in a field-decay convention and gives an intensity full width \(2(\kappa_{\rm ex}+\kappa_{\rm in})\) for the add/drop response.  The corresponding intrinsic intensity linewidth is therefore \(\Gamma_{\rm int}\approx2\) GHz and \(Q_{\rm int}\approx9.7\times10^4\) near \(1550\) nm.  We adopt the explicit design target \(\Gamma_{\rm loaded}\le2.2\) GHz, or \(Q_{\rm loaded}\gtrsim8.8\times10^4\), leaving only about \(0.2\) GHz of total linewidth for all readout ports and additional loss.  At this loaded linewidth the \(1.922\) GHz doublet remains two-peaked by a factor \(3.33/2.2=1.51\), while the square and triangular shell gaps are about \(7.9\) and \(8.2\) linewidths.

The robust visual multiplicity count uses a larger generic detuning, because \eqref{eq:square-device-branches} is exact in \(\varepsilon\).  At \((\varepsilon_u,\varepsilon_v)=(0.25,0.125)\) the four lines are \(1.993,4.393,6.771,9.171\) GHz.  Their minimum adjacent separation is \(2.377\) GHz, and the equal-residue four-Lorentzian sum retains four maxima up to \(\Gamma\approx3.30\) GHz, a factor \(1.50\) above the \(2.2\) GHz target.  We have also checked that these four branches remain the four lowest throughout the scan region used below.  The multiplicity count should therefore be performed at large detuning and the slope fit at small detuning, where linearity is the point.  Unequal residues and coherent interference can shift visual thresholds, which is why the multiport pole fit remains the primary analysis.  The harder centroid measurements predict coordinate curvatures \(C_8=18.237\) GHz for the square graph and \(36.819\) GHz for the triangular graph; at \(\varepsilon_u=0.10\) their exact centroid shifts are only \(0.091\) and \(0.184\) GHz, well inside a linewidth, so they rely on a global multi-line fit rather than on resolving a displacement.

\subsection{Calibration tolerances and measurement sequence}

The ideal degeneracies will be obscured unless the site resonances and couplings are calibrated, and the distance from present practice should be stated numerically rather than as an aspiration.  \Cref{tab:tolerances}, the fourth and final table, compares the scales quoted for the platform of Ref.~\cite{HafeziEtAl2013} with the targets adopted here.  The shell-spread columns come from 200 realizations of the \(64\times64\) square matrix \eqref{eq:fit-hamiltonian}, with independent zero-mean Gaussian site and bond errors and a fixed random seed.  They are a reproducible design Monte Carlo, not measurements or universal fabrication limits.

Two points follow.  First, mapping the reported characteristic site mismatch \(0.8j\) to a Gaussian root-mean-square scale gives a shell spread of order \(7\) GHz, so the experiment would fail without trimming.  Independent heaters or permanent trimming of all \(64\) sites are therefore part of the device, and the site-uniformity scale must improve by roughly a factor of \(60\).  Second, site and bond disorder are both self-averaged in delocalized Bloch states, but their prefactors differ: a site matrix element scales as \(\sigma_\nu/N\), whereas bond disorder involves more matrix elements with absolute scale \(j\sigma_j/j\).  In the Monte Carlo, \(0.2\) GHz site disorder alone gives a mean shell spread \(0.11\) GHz, while \(1\%\) hopping disorder alone gives \(0.24\) GHz.  We therefore adopt \(\sigma_j/j\lesssim0.5\%\) as a prudent design target, yielding a combined mean spread \(0.16\) GHz and a 90th percentile \(0.23\) GHz.  This is a simulation-based target to revisit with foundry statistics, not a sharp physical threshold.  It matters especially for the seam links, whose geometry differs most from the bulk; the triangular diagonal links must meet the same tolerance.

\begin{table}[t]
\centering
\caption{Calibration budget.  The reported characteristic scales refer to the ring-array platform of Ref.~\cite{HafeziEtAl2013} (\(j\approx16\) GHz).  The last two columns give the mean and 90th-percentile spread of the fourfold shell at \((u,v)=(1/2,1/2)\) in the stated Gaussian Monte Carlo.  Backscattering is not included in this single-circulation \(64\times64\) simulation.}
\label{tab:tolerances}
\scriptsize
\begin{tabular}{@{}l c c c c@{}}
\toprule
Disorder channel or scenario & Reported & Target & Mean spread & 90th percentile \\
\midrule
Site frequency \(\sigma_\nu\)      & \(0.8j\approx12.8\) GHz & \(\lesssim0.2\) GHz  & --- & ---\\
Hopping mismatch \(\sigma_j/j\)    & \(0.04\)                & \(\lesssim0.005\)    & --- & ---\\
Backscattering \(\beta\)           & \(0.04j\approx0.64\) GHz& \(\lesssim0.2\) GHz  & --- & ---\\
\midrule
\multicolumn{3}{@{}l}{reported scales mapped to Gaussian RMS (\(12.8\) GHz, \(4\%\))} & \(7.45\) GHz & \(10.57\) GHz\\
\multicolumn{3}{@{}l}{sites trimmed only (\(0.2\) GHz, \(4\%\))}                    & \(0.98\) GHz & \(1.37\) GHz\\
\multicolumn{3}{@{}l}{previous budget (\(0.2\) GHz, \(1\%\))}                       & \(0.27\) GHz & \(0.37\) GHz\\
\multicolumn{3}{@{}l}{budget adopted here (\(0.2\) GHz, \(0.5\%\))}                 & \(0.16\) GHz & \(0.23\) GHz\\
\bottomrule
\end{tabular}
\end{table}

After projection to the selected circulation, loss and residual detuning are included in the actual fit through the non-Hermitian matrix
\begin{equation}
 H_{\rm fit}/h=\nu_{\rm ref}I+jL(u,v)+\diag(\delta\nu_r)-\frac{\ii}{2}\diag(\Gamma_r).
 \label{eq:fit-hamiltonian}
\end{equation}

The classical test is then operationally fixed:
\begin{enumerate}
 \item calibrate \(\nu_r\), \(j\), and \(\Gamma_r\) from the untwisted spectrum and tune the rings into the target tolerance;
 \item establish the fourfold multiplicity first, at a \emph{large} generic detuning where the branches are widely separated: \((\varepsilon_u,\varepsilon_v)=(0.25,0.125)\) gives lines at \(1.993,4.393,6.771,9.171\) GHz and a smallest adjacent separation of \(2.377\) GHz.  For equal incoherent residues the full four-line sum retains four maxima up to \(\Gamma\approx3.30\) GHz, a factor \(1.50\) above the \(2.2\) GHz loaded-linewidth target; confirm the count independently from the fitted pole-residue rank;
 \item verify the \(17.318\) GHz next-shell gap, then sweep \(u\in[0.35,0.65]\) at fixed \(v=1/2\) in steps of \(0.02\), fitting all poles of the transmission matrix rather than only intensity maxima, and check the exact branches \eqref{eq:square-device-branches} and the cusp slope \(A_8=9.618\) GHz.  Note that at \(v=1/2\) exactly the \(t=\pm1\) branches coincide, so this scan shows a two-line splitting; the fourfold count comes from the previous step, not from this one;
 \item repeat at the triangular point, verifying the threefold crossing and the three frequencies in \Cref{tab:device-numbers};
 \item shear one coupling family or detune one diagonal-link family and confirm the predicted loss of degeneracy.
\end{enumerate}
The phase scan corresponds to changing the seam phase by \(2\pi\varepsilon_u\); the \(0.10\) point is therefore a \(36^\circ\) offset from the symmetry-fixed value.  The primary observables are the multiplicity, shell gap and linear splitting.  The centroid curvature is secondary because its predicted shift is well below the linewidth, although a global multi-line fit constrained by the exact branch formula can locate a centroid far more accurately than one linewidth.

None of these measurements detects complex multiplication or the determinant response.  They test the finite graph regularization of the lattice spectrum and its symmetry-fixed deep holes, exactly as delimited in \Cref{rem:no-signature}.

\subsection{Quantum stage}

Only after the classical spectrum is fitted should the logical proposal be tested.  In the square case one introduces and characterizes the selector, then injects a heralded single photon into a superposition of the selected pair.  In the triangular case one prepares the three resolved mode amplitudes directly.  Leakage is measured while the seam phase is scanned.  A successful classical spectrum does not by itself demonstrate a qubit or qutrit; the single-photon restriction, coherent preparation, tomography, a noncommuting intra-manifold coupler and, eventually, two-particle interference remain separate requirements.

\section{Discussion}

Four statements are often run together and should be kept apart.
\begin{itemize}
 \item \emph{Deep hole \(=\) gap optimum} is a definition.
 \item \emph{CM \(\Rightarrow\) torsion optimum} is a correct theorem (\Cref{prop:rational}), but one-directional and not inferable from the optimum's location (\Cref{prop:rectangular,rem:not-detectable}).
 \item \emph{Extra unit \(\Rightarrow\) the isotropic response \(\mathcal K=2\pi(\Im\tau)G_\tau^{-1}\)} is exact (\Cref{thm:units,thm:response}) and fails without the unit (\Cref{rem:not-general}).  The same symmetry forces the current to vanish at the hexagonal deep holes, whereas at a two-torsion point current cancellation already follows from inversion.  The exceptional tensor form is the arithmetic--geometric statement: what one would see, if one could see it, is enhanced rotational symmetry, whose classification through the unit group is the arithmetic interpretation.
 \item \emph{An arithmetic experimental signature} is \emph{not} established (\Cref{rem:no-signature}).
\end{itemize}

The coordination number of the deep hole controls the dimension of the lowest manifold: four for the square lattice, three for the hexagonal.  This is the cleanest link between the geometry and the logical dimension, and it is worth pursuing in higher-dimensional synthetic lattices, where the deep holes of \(D_4\), \(E_8\) and the Leech lattice have well-understood coordination \cite{ConwaySloane1999}.

The finite-device calculation changes the experimental status of the paper but not its arithmetic status.  Equations \eqref{eq:square-device-spectrum} and \eqref{eq:triangular-device-spectrum} give falsifiable GHz-scale line positions without invoking the determinant.  They also expose the actual bottleneck: not spectral resolution of the shell gap, which is generous, but resonance alignment across 64 rings and controlled closure of the two seams.  A referee can therefore distinguish a failed physical design from a failed lattice prediction.

Limitations.  The round result is a scalar proxy in a specific self-adjoint realization; a fabricated structure requires the vector Maxwell operator with material boundary conditions.  The microring proposal is a tight-binding coupled-mode design, not that vector solution, and \eqref{eq:fit-hamiltonian} treats loss and disorder phenomenologically.  The wrap-around routing, 64-ring frequency alignment, and equal triangular couplings have not been demonstrated together.  The modulus \(\tau\) is not determined by the Hopf geometry.  The holonomies commute, so universal control needs an extra element.  The next calculations are a layout-level electromagnetic simulation, Monte Carlo disorder and thermal-crosstalk analysis, and optimization of the input/output coupling matrix.

\section{Conclusions}

Two commuting holonomies on a two-cycle photonic geometry have spectrally distinguished operating points at the deep holes of the momentum lattice.  Their location follows from lattice symmetry and carries no arithmetic information; complex multiplication guarantees torsion optima but cannot be inferred from them.  The structural role of arithmetic is the unit group: at \(\tau=\ii\) and \(\tau=\rho\) the quarter-turn and the \(120^\circ\) rotation fix exactly the origin and the deep holes.  At these points the current vanishes; at the square two-torsion point this already follows from inversion, whereas at the hexagonal deep holes it follows from the order-three unit.  What the extra rotational symmetry forces in both cases, together with \(\Delta_\zeta\log D_\tau=-4\pi/\Im\tau\), is the exact susceptibility \(\mathcal K=2\pi(\Im\tau)G_\tau^{-1}\), a tensor form that demonstrably fails without the extra unit.

The continuum data are: a fourfold manifold with gap \(2\kappa^2\) and \(D_\ii(1/2,1/2)=2\) at \(\tau=\ii\); a threefold manifold with gap \(\tfrac43\kappa^2\) and \(D_\rho(1/3,2/3)=\sqrt3\) at \(\tau=\rho\); a linear splitting with stationary centroid of curvature \(2G_\tau^{-1}\).  The proposed \(8\times8\) microring regularization converts these statements into exact transmission-spectrum targets.  At \(j=16\) GHz it predicts gaps of \(17.318\) GHz and \(18.106\) GHz and a \(1.922\) GHz square doublet separation at a \(0.1\) phase-coordinate error, against a design target \(\Gamma_{\rm loaded}\le2.2\) GHz.  A larger generic detuning gives four visibly separated branches in the equal-residue model.  This is now a quantitative experimental proposal, although its seam closure, calibration, and signal-to-noise budget remain to be demonstrated.

An isolated manifold is not a logical system: after restriction to the single-photon sector, the square case yields a dual-rail qubit once a selector removes two of the four modes, while the triangular case yields a qutrit without a dimension-reducing selector.  Leakage from the parent manifold is bounded by the measured shell gap; mode-resolved preparation, intra-manifold dephasing, universal control and entanglement remain engineering problems.

\section*{Author Contributions}
Conceptualization, methodology, formal analysis, software, writing---original draft and writing---review and editing: M.P.  The author has read and agreed to the published version of the manuscript.

\section*{Funding}
This research received no external funding.

\section*{Conflicts of Interest}
The author declares no conflict of interest.

\section*{Use of Artificial Intelligence}
Large language models were used for manuscript restructuring, algebraic cross-checks, adversarial review, and code-assisted numerical verification.  All mathematical statements were independently verified by the author, who takes full responsibility for the scientific content.

\appendix

\section{Shell calculations}

\emph{Square.}  At \(a=(1/2,1/2)\), \(\tau=\ii\): the four \((m,n)\in\{-1,0\}^2\) give \(Q=1/2\); the next eight have one component of magnitude \(3/2\), giving \(Q=5/2\); the third shell is fourfold at \(9/2\).

\emph{Hexagonal.}  With \(Q_\rho(q)=\tfrac43(q_1^2-q_1q_2+q_2^2)\) and \(a=(1/3,2/3)\), the offsets \((0,0),(0,-1),(-1,-1)\) give \(Q=4/9\); the offsets \((-1,-2),(-1,0),(1,0)\) give \(16/9\); the third shell is sixfold at \(28/9\).  The gap is \(\tfrac43\).

\section{Numerical verifications}

All computations use \(30\)-digit arithmetic unless stated.

\emph{Fixed sets.}  \(\det(I-S_4)=2\), \(\det(I-S_3)=3\), \(\det(I-S_6)=1\); the computed fixed sets match \Cref{thm:units}(ii) and \Cref{rem:delicate}.

\emph{Laplacian identity.}  For \(\tau\in\{\ii,\rho,2\ii,0.3+1.1\ii\}\), \(\operatorname{tr}\Hess_\zeta\log D\) agrees with \(-4\pi/\Im\tau\) to ten digits, confirming \eqref{eq:laplacian} including at non-exceptional moduli.

\emph{Determinant values.}  \(D_\ii(1/2,1/2)=2\) and \(\vartheta_3(0\mid\ii)/\sqrt2=\eta(\ii)=0.768225422326056659\), confirming \Cref{prop:square-value}.  For the degree-three self-isogeny \(\phi=I-S_3\), numerical evaluation confirms
\(D_\rho(\phi a)=\prod_{t\in\ker\phi}D_\rho(a+t)\) at generic \(a\), and the product over its two nonzero kernel points is \(3\), confirming \Cref{prop:hex-value}.  The eight nonzero \(3\)-torsion values are \(\sqrt3\) twice and \(3^{1/6}=1.200936955176\) six times, with product \(9.0\).

\emph{Validity domains.}  The four branches of \Cref{prop:splitting} are the four lowest exactly when \eqref{eq:sharp-domain} holds; e.g.\ it fails at \(\varepsilon=(0.34,0.34)\) and holds at \((0.33,0.33)\).  In the round model the boundary is \(|\varepsilon_u|+|\varepsilon_v|=1/2\); e.g.\ it fails at \((0.26,0.26)\) and holds at \((0.24,0.24)\).

\emph{Finite device.}  Direct diagonalization of the \(8\times8\) twisted square and triangular graph Laplacians agrees with \eqref{eq:square-device-spectrum} and \eqref{eq:triangular-device-spectrum}.  At \(j=16\) GHz it gives the multiplicities, gaps and \(\varepsilon_u=0.10\) line positions of \Cref{tab:device-numbers}.  Finite differences at the deep holes give \(A_8=9.617883678\), \(C_8=18.236651000\), and \(B_8=19.071202700\) GHz per holonomy coordinate, and the exact centroid shifts at \(\varepsilon_u=0.10\) are \(0.09114\) GHz and \(0.18400\) GHz.  The large generic detuning \((0.25,0.125)\) gives \(1.992780284,4.393390568,6.770881663,9.171491947\) GHz.  A direct four-Lorentzian calculation gives a four-to-two-maxima transition near \(\Gamma=3.30\) GHz for equal residues, whereas the small generic quartet already has only two maxima at \(\Gamma=2\) GHz.  The Gaussian-disorder Monte Carlo values, including the 90th percentiles in \Cref{tab:tolerances}, use 200 realizations and a fixed seed.


\small
\begin{thebibliography}{99}

\bibitem{KokEtAl2007}
Kok, P.; Munro, W.J.; Nemoto, K.; Ralph, T.C.; Dowling, J.P.; Milburn, G.J.
Linear optical quantum computing with photonic qubits.
\emph{Rev. Mod. Phys.} \textbf{2007}, \emph{79}, 135--174.
\url{https://doi.org/10.1103/RevModPhys.79.135}.

\bibitem{BrechtEtAl2015}
Brecht, B.; Reddy, D.V.; Silberhorn, C.; Raymer, M.G.
Photon temporal modes: A complete framework for quantum information science.
\emph{Phys. Rev. X} \textbf{2015}, \emph{5}, 041017.
\url{https://doi.org/10.1103/PhysRevX.5.041017}.

\bibitem{FabreTreps2020}
Fabre, C.; Treps, N.
Modes and states in quantum optics.
\emph{Rev. Mod. Phys.} \textbf{2020}, \emph{92}, 035005.
\url{https://doi.org/10.1103/RevModPhys.92.035005}.

\bibitem{LuEtAl2020}
Lu, L.; Xia, L.; Chen, Z.; Chen, L.; Yu, T.; Tao, T.; Ma, W.; Pan, Y.; Cai, X.; Lu, Y.; Zhu, S.; Ma, X.-S.
Three-dimensional entanglement on a silicon chip.
\emph{npj Quantum Inf.} \textbf{2020}, \emph{6}, 30.
\url{https://doi.org/10.1038/s41534-020-0260-x}.

\bibitem{FangYuFan2012}
Fang, K.; Yu, Z.; Fan, S.
Realizing effective magnetic field for photons by controlling the phase of dynamic modulation.
\emph{Nat. Photonics} \textbf{2012}, \emph{6}, 782--787.
\url{https://doi.org/10.1038/nphoton.2012.236}.

\bibitem{MittalEtAl2014}
Mittal, S.; Fan, J.; Faez, S.; Migdall, A.; Taylor, J.M.; Hafezi, M.
Topologically robust transport of photons in a synthetic gauge field.
\emph{Phys. Rev. Lett.} \textbf{2014}, \emph{113}, 087403.
\url{https://doi.org/10.1103/PhysRevLett.113.087403}.

\bibitem{YuanShiFan2016}
Yuan, L.; Shi, Y.; Fan, S.
Photonic gauge potential in a system with a synthetic frequency dimension.
\emph{Opt. Lett.} \textbf{2016}, \emph{41}, 741--744.
\url{https://doi.org/10.1364/OL.41.000741}.

\bibitem{DuttEtAl2020}
Dutt, A.; Lin, Q.; Yuan, L.; Minkov, M.; Xiao, M.; Fan, S.
A single photonic cavity with two independent physical synthetic dimensions.
\emph{Science} \textbf{2020}, \emph{367}, 59--64.
\url{https://doi.org/10.1126/science.aaz3071}.

\bibitem{OzawaEtAl2019}
Ozawa, T.; Price, H.M.; Amo, A.; Goldman, N.; Hafezi, M.; Lu, L.; Rechtsman, M.C.;
Schuster, D.; Simon, J.; Zilberberg, O.; Carusotto, I.
Topological photonics.
\emph{Rev. Mod. Phys.} \textbf{2019}, \emph{91}, 015006.
\url{https://doi.org/10.1103/RevModPhys.91.015006}.

\bibitem{HafeziEtAl2011}
Hafezi, M.; Demler, E.A.; Lukin, M.D.; Taylor, J.M.
Robust optical delay lines with topological protection.
\emph{Nat. Phys.} \textbf{2011}, \emph{7}, 907--912.
\url{https://doi.org/10.1038/nphys2063}.

\bibitem{HafeziEtAl2013}
Hafezi, M.; Mittal, S.; Fan, J.; Migdall, A.; Taylor, J.M.
Imaging topological edge states in silicon photonics.
\emph{Nat. Photonics} \textbf{2013}, \emph{7}, 1001--1005.
\url{https://doi.org/10.1038/nphoton.2013.274}.

\bibitem{DinhEtAl2024}
Dinh, H.X.; Bal\v{c}ytis, A.; Ozawa, T.; Ota, Y.; Ren, G.; Baba, T.; Iwamoto, S.; Mitchell, A.; Nguyen, T.G.
Reconfigurable synthetic dimension frequency lattices in an integrated lithium niobate ring cavity.
\emph{Commun. Phys.} \textbf{2024}, \emph{7}, 185.
\url{https://doi.org/10.1038/s42005-024-01676-9}.

\bibitem{Hofstadter1976}
Hofstadter, D.R.
Energy levels and wave functions of Bloch electrons in rational and irrational magnetic fields.
\emph{Phys. Rev. B} \textbf{1976}, \emph{14}, 2239--2249.
\url{https://doi.org/10.1103/PhysRevB.14.2239}.

\bibitem{AharonovBohm1959}
Aharonov, Y.; Bohm, D.
Significance of electromagnetic potentials in the quantum theory.
\emph{Phys. Rev.} \textbf{1959}, \emph{115}, 485--491.
\url{https://doi.org/10.1103/PhysRev.115.485}.

\bibitem{ButtikerImryLandauer1983}
B\"uttiker, M.; Imry, Y.; Landauer, R.
Josephson behavior in small normal one-dimensional rings.
\emph{Phys. Lett. A} \textbf{1983}, \emph{96}, 365--367.
\url{https://doi.org/10.1016/0375-9601(83)90011-7}.

\bibitem{Bar2019}
B\"ar, C.
The curl operator on odd-dimensional manifolds.
\emph{J. Math. Phys.} \textbf{2019}, \emph{60}, 031501.
\url{https://doi.org/10.1063/1.5082528}.

\bibitem{Folland1989}
Folland, G.B.
Harmonic analysis of the de Rham complex on the sphere.
\emph{J. Reine Angew. Math.} \textbf{1989}, \emph{398}, 130--143.
\url{https://doi.org/10.1515/crll.1989.398.130}.

\bibitem{IrvineBouwmeester2008}
Irvine, W.T.M.; Bouwmeester, D.
Linked and knotted beams of light.
\emph{Nat. Phys.} \textbf{2008}, \emph{4}, 716--720.
\url{https://doi.org/10.1038/nphys1056}.

\bibitem{KediaEtAl2013}
Kedia, H.; Bia{\l}ynicki-Birula, I.; Peralta-Salas, D.; Irvine, W.T.M.
Tying knots in light fields.
\emph{Phys. Rev. Lett.} \textbf{2013}, \emph{111}, 150404.
\url{https://doi.org/10.1103/PhysRevLett.111.150404}.

\bibitem{KopinskiNatario2017}
Kopi\'nski, J.; Nat\'ario, J.
On a remarkable electromagnetic field in the Einstein Universe.
\emph{Gen. Relativ. Gravit.} \textbf{2017}, \emph{49}, 81.
\url{https://doi.org/10.1007/s10714-017-2242-7}.

\bibitem{ConwaySloane1999}
Conway, J.H.; Sloane, N.J.A.
\emph{Sphere Packings, Lattices and Groups}, 3rd ed.; Grundlehren der mathematischen Wissenschaften;
Springer: New York, NY, USA, 1999; Volume 290.
\url{https://doi.org/10.1007/978-1-4757-6568-7}.

\bibitem{Silverman2009}
Silverman, J.H.
\emph{The Arithmetic of Elliptic Curves}, 2nd ed.; Graduate Texts in Mathematics;
Springer: New York, NY, USA, 2009; Volume 106.
\url{https://doi.org/10.1007/978-0-387-09494-6}.

\bibitem{Siegel1980}
Siegel, C.L.
\emph{On Advanced Analytic Number Theory}, 2nd ed.; Tata Institute of Fundamental Research:
Bombay, India, 1980.

\bibitem{Lang1987}
Lang, S.
\emph{Elliptic Functions}, 2nd ed.; Graduate Texts in Mathematics;
Springer: New York, NY, USA, 1987; Volume 112.
\url{https://doi.org/10.1007/978-1-4612-4752-4}.

\bibitem{RaySinger1971}
Ray, D.B.; Singer, I.M.
\(R\)-torsion and the Laplacian on Riemannian manifolds.
\emph{Adv. Math.} \textbf{1971}, \emph{7}, 145--210.
\url{https://doi.org/10.1016/0001-8708(71)90045-4}.

\bibitem{Quillen1985}
Quillen, D.
Determinants of Cauchy--Riemann operators over a Riemann surface.
\emph{Funct. Anal. Appl.} \textbf{1985}, \emph{19}, 31--34.
\url{https://doi.org/10.1007/BF01086022}.

\bibitem{KubertLang1981}
Kubert, D.S.; Lang, S.
\emph{Modular Units}; Grundlehren der mathematischen Wissenschaften;
Springer: New York, NY, USA, 1981; Volume 244.
\url{https://doi.org/10.1007/978-1-4757-1741-9}.

\bibitem{deJong2005}
de Jong, R.
On the Arakelov theory of elliptic curves.
\emph{Enseign. Math.} \textbf{2005}, \emph{51}, 179--201.
\href{https://arxiv.org/abs/math/0312359}{arXiv:math/0312359}.

\bibitem{PlanatEtAl2022Character}
Planat, M.; Amaral, M.M.; Fang, F.; Chester, D.; Aschheim, R.; Irwin, K.
Character varieties and algebraic surfaces for the topology of quantum computing.
\emph{Symmetry} \textbf{2022}, \emph{14}, 915.
\url{https://doi.org/10.3390/sym14050915}.

\bibitem{PlanatEtAl2022Fricke}
Planat, M.; Chester, D.; Amaral, M.M.; Irwin, K.
Fricke topological qubits.
\emph{Quantum Rep.} \textbf{2022}, \emph{4}, 523--532.
\url{https://doi.org/10.3390/quantum4040037}.

\end{thebibliography}
\end{document}